\documentclass[12pt,a4paper]{article}
\usepackage[english]{babel}
\usepackage{float}
\usepackage{amsmath, amssymb, amsthm, epsfig, color, mathtools, bbm, dsfont, mathrsfs, bm}
\usepackage[latin1]{inputenc}
\usepackage[T1]{fontenc}
\usepackage{indentfirst}
\usepackage[shortlabels]{enumitem}
\usepackage[title]{appendix}
\usepackage{lmodern}       
\usepackage{comment}  
\usepackage{textcomp, upgreek}
\usepackage{setspace}
\usepackage{soul}
\usepackage{accents}
\usepackage{hyperref} 
\usepackage{geometry}
\usepackage{xcolor}

\usepackage{xr}    
\usepackage{graphicx} 
\usepackage[export]{adjustbox} 
\usepackage{subcaption} 
\usepackage{mathtools} 
\usepackage{siunitx}
\usepackage{bm} 

\usepackage{addfont}
\addfont{OT1}{rsfs10}{\rsfs}

\usepackage{tikz, tkz-euclide}
\usetikzlibrary{patterns}
\usetikzlibrary{arrows.meta}
\usetikzlibrary{calc}
\usetikzlibrary{decorations.markings}
\usepackage{fancybox}

\usepackage{graphicx}
\graphicspath{ {./images/} }
\newcommand{\Z}{\mathbb{Z}}
\newcommand{\E}{\mathbb{E}}
\newcommand{\V}{\mathbb{V}}

\newcommand{\calW}{\mathcal{W}}
     
\newcommand{\supp}{\text{supp}\;}

\def\H+{{\mathbb{H}^d_+}}

\newcommand{\rmh}{\mathrm{h}}
\newcommand{\lambdadec}{\widehat{\bm{\lambda}}}

\def\supp{\mathop{\textrm{\rm supp}}\nolimits}            
\def\d{\mathop{\textrm{\rm d}}\nolimits}                  

\newcommand{\be}{\begin{equation}}
\newcommand{\ee}{\end{equation}}

\numberwithin{equation}{section}

  \newcounter{dummy} \numberwithin{dummy}{section}
  \theoremstyle{plain}
  \newtheorem*{theorem*}        {Theorem}
	\newtheorem*{conjecture*}   {Conjecture}
  \newtheorem{theorem}[dummy]          {Theorem}
  \newtheorem{lemma}[dummy]              {Lemma}
  \newtheorem*{lemma*}          {Lemma}

  \newtheorem{proposition}[dummy]       {Proposition}
   
  \newtheorem{remark}[dummy]           {Remark}
  
  \theoremstyle{remark}
  \theoremstyle{definition}
   
\makeatletter

\newcommand\longleftrightarrowfill@{%
  \arrowfill@\leftarrow\relbar\rightarrow}
\makeatother

\definecolor{Red}{cmyk}{0,1,1,0}

\definecolor{Blue}{cmyk}{1,1,0,0}

\definecolor{DarkBlue}{rgb}{0.1,0.1,0.5}
\definecolor{Red}{rgb}{0.9,0.0,0.1}
\definecolor{DarkGreen}{rgb}{0.10,0.50,0.10}
\definecolor{DarkRed}{rgb}{0.50,0.10,0.10}
\definecolor{bleu}{RGB}{0,140,189}%

\newgeometry{vmargin={15mm}, hmargin={12mm,17mm}}

\definecolor{vermelho}{RGB}{208,2,27}    
\definecolor{verde}{RGB}{126,211,33} 

\begin{document}

\begin{center}
{\LARGE Phase Transitions in the semi-infinite Ising model with a decaying field}
\vskip.5cm
Rodrigo Bissacot and Jo{\~a}o Maia
\vskip.3cm
\begin{footnotesize}
Institute of Mathematics and Statistics (IME-USP), University of S\~{a}o Paulo, Brazil\\
\end{footnotesize}
\vskip.1cm
\begin{scriptsize}
emails: rodrigo.bissacot@gmail.com, maia.joaovt@gmail.com
\end{scriptsize}

\end{center}

\begin{abstract}
We study the semi-infinite Ising model with an external field $h_i = \lambda |i_d|^{-\delta}$, $\lambda$ is the wall influence, and $\delta>0$. This external field decays as it gets further away from the wall. We are able to show that when $\delta>1$ and $\beta > \beta_c(d)$, there exists a critical value $0< \lambda_c:=\lambda_c(\delta,\beta)$ such that, for $\lambda<\lambda_c$ there is phase transition and for $\lambda>\lambda_c$ we have uniqueness of the Gibbs state. In addition, when $\delta<1$ we have only one Gibbs state for any positive $\beta$ and $\lambda$.
\end{abstract}

\section{Introduction}The semi-infinite Ising model is a variation of the Ising model where, instead of $\Z^d$($d\geq 2$), the lattice is $\H+=\mathbb{Z}^{d-1}\times\mathbb{N}$ and the configurations space is $\Omega \coloneqq\{-1,+1\}^{\H+}$. In the semi-infinite model, the sites in the wall $\mathcal{W}=\mathbb{Z}^{d-1}\times \{1\}$ are in contact with a substrate favoring one of the spins. This influence is represented by an external field, with intensity $\lambda\in\mathbb{R}$, acting only on spins at the wall $\calW$. 
The other parameters of the model are the interaction $\bm{J}=(J_{i,j})_{i,j\in\H+}$, the external field $\bm{h}=(h_i)_{i\in\H+}$ and the inverse temperature $\beta$. We will always consider nearest neighbor ferromagnetic interaction, hence $J_{i,j}= J > 0$ whenever $|i-j|=1$ and are zero otherwise. The distance here is taken concerning the $\ell_1$-norm. The interaction $\bm{J}$ and the external field $\bm{h}$ play the same role in the energy as in the standard Ising model, so the formal Hamiltonian is 
\begin{equation*}
H_{J,\lambda, \bm{h}}(\sigma) = -\sum_{\substack{i,j\in\H+ \\ |i-j|=1}}J\sigma_i\sigma_j - \sum_{i\in\H+} h_i\sigma_i - \sum_{i\in\mathcal{W}}\lambda\sigma_i.
\end{equation*}
The role of the temperature is expressed in the formal Gibbs measure, given by
\begin{equation*}
\mu_{\bm{J},\lambda, \bm{h}}^{\beta}(\sigma)=\frac{e^{-\beta H_{J,\lambda,\bm{h}}(\sigma)}}{\mathcal{Z}^{\beta}_{\bm{J}, \lambda, \bm{h}}}     
\end{equation*}
where $\mathcal{Z}^{\beta}_{\bm{J}, \lambda, \bm{h}}$ is the partition function, a normalizing weight. Throughout this paper, we will assume $\lambda\geq0$ for simplicity. The extension of our statements for $\lambda\leq 0$ will follow from spin-flip symmetry.

This model was extensively studied by Fr{\"o}hlich and Pfister in \cite{FP-I, FP-II}, where they presented a large range of results for the model. Regarding the macroscopic behaviour of the system, it was shown that, for $\beta>\beta_c$, the system behaves exactly as the Ising model and the spins align independently, where $\beta_c$ denotes the critical inverse temperature of the Ising model. Below this critical temperature, when there is no external field, there exists a critical value $\lambda_c>0$ that determines the behavior of the spins near the wall. When the influence of the wall is bigger than $\lambda_c$, the influence of the substrate "penetrates" the model and we see a thick layer of spins near the wall aligned with the substrate phase:

\begin{figure}[htbp]
    \centering
    \includegraphics[scale=0.4]{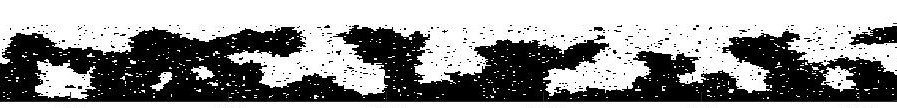}
    \caption{Complete wetting: $\lambda = 1 > \lambda_c$, $\beta=0.5$, $J=1$. The $+$ spins are black and the $-$ are white.}
\end{figure}
This regime is then called \textit{complete wetting}. When $0\leq \lambda<\lambda_c$, the influence of the substrate is only capable of creating disconnected clusters on the wall, so we say there is \textit{partial wetting}:

\begin{figure}[H]
    \centering
    \includegraphics[scale=0.5]{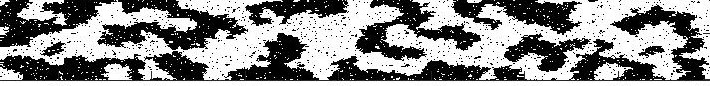}
    \caption{Partial wetting: $\lambda = 0.03 <\lambda_c$, $\beta=0.5$, $J=1$. The $+$ spins are black and the $-$ are white.}
\end{figure}
 In \cite{FP-II}, they also showed that this critical value is related to the existence of multiple Gibbs states, in the sense that, for $0\leq\lambda<\lambda_c$ there are multiple Gibbs states, and for $\lambda>\lambda_c$ we have uniqueness. The existence of this critical value $\lambda_c$ is proved using a notion of wall-free energy, defined formally as 
\begin{equation*}
    \tau_w(J,\lambda, \bm{h}) = \lim_{\Lambda\to\Z^d} \frac{1}{2|\mathcal{W}\cap\Lambda|}\ln\left[ \frac{(\mathcal{Z}^{-, J}_{\Lambda\cap\H+; \lambda,\bm{h}})^2}{Q^{-, J}_{\Lambda; \bm{h}}} \right] - \lim_{\Lambda\to\Z^d}\frac{1}{2|\mathcal{W}\cap \Lambda|}\ln\left[ \frac{(\mathcal{Z}^{+, J}_{\Lambda\cap \H+; \lambda,\bm{h}})^2}{Q^{+, J}_{\Lambda; \bm{h}}} \right],
\end{equation*}
where $Q^{\pm, J}_{\Lambda; \bm{h}}$ is the partition function of the Ising model (on $\Z^d$) with $\pm$ boundary condition on $\Lambda$, a finite set such that $\mathcal{W}\cap\Lambda \neq \emptyset$. Here $\mathcal{Z}^{\pm, J}_{\Lambda\cap\H+; \lambda,\bm{h}}$ is the partition function of the semi-infinite Ising model with $\pm$ boundary condition on $\Lambda\cap\H+$, considering the free boundary condition on $\mathcal{W}$. This quantity is suitable for measuring the wall influence when there is no external field since the partition functions of the Ising model cancel out, and the remaining terms can be written as an integral of the difference of the magnetization with respect to the wall influence $\lambda$, see Proposition \ref{Prop: Tilde_tau_as_integral}. 

\begin{figure}[H]
\centering
\includegraphics[scale=0.2]{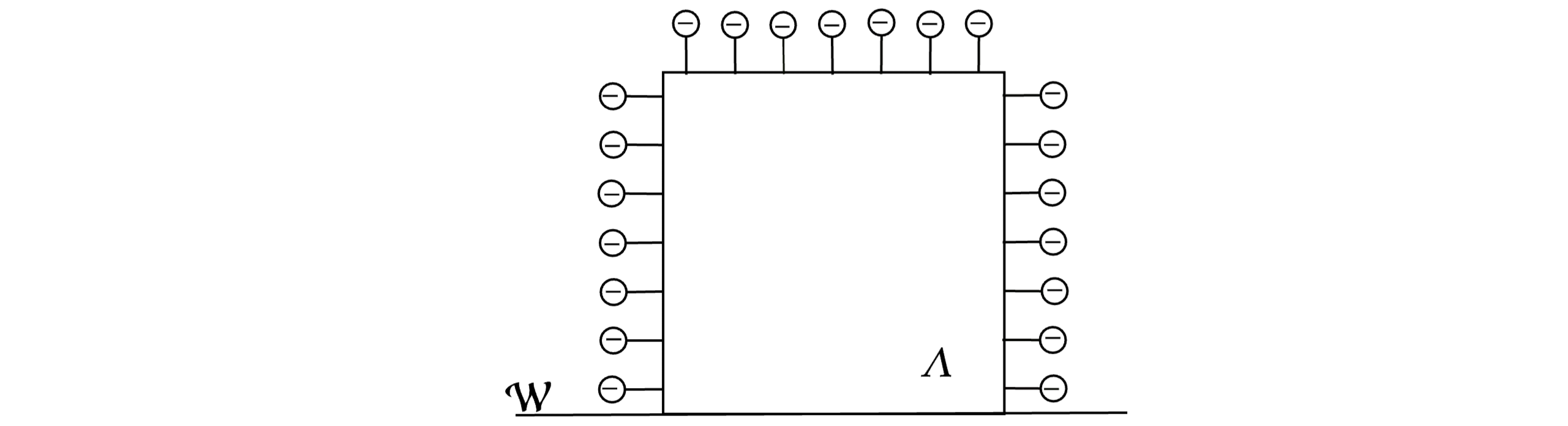}
\caption{The \textcolor{green}{-} boundary condition of the semi-infinite model.}
\end{figure}

In the Ising model, adding a nonnull constant external field disrupts the phase transition at every temperature, as a consequence of Lee-Yang theorem \cite{FV-Book, Lee-Yang.II.1952}. However, it was shown in \cite{Bissacot_Cioletti_10} that we can add an external field that decays as it goes to infinity and still preserves phase transition. This work started a streak of new results on models with decaying fields \cite{Affonso.2021, Bissacot_Cass_Cio_Pres_15, Bissacot_Endo_Enter_2017, Bissacot.Endo.18}.

One particular result \cite{Bissacot_Cass_Cio_Pres_15}, states that we can consider an intermediate external field $\bm{h}^* = (h_i^*)_{i\in\Z^d}$ given by
\begin{equation}\label{particular.external.field}
    h_i^* = \begin{cases}
            h^* &\text{ if }i=0,\\
            \frac{h^*}{|i|^\delta} &\text{ otherwise}.\\
            \end{cases}
\end{equation}
that preserves phase transition for low temperatures when $\delta>1$ and induces uniqueness at low temperatures when $\delta<1$. In the critical value $\delta=1$, there is phase transition for $h^*$ small enough. The proof of uniqueness when $\delta<1$ was extended to all temperatures in \cite{Cioletti_Vila_2016}. The argument in \cite{Bissacot_Cass_Cio_Pres_15} involves contour arguments and Peierls' bounds techniques for low temperatures, while \cite{Cioletti_Vila_2016} uses a generalization of the Edwards-Sokal representation. Both techniques are fairly distinct and complement each other, which makes the complete proof of uniqueness involved. There is no standard strategy to prove uniqueness. 

For the semi-infinite Ising model, a more natural choice of the external field is one decaying as it gets further from the wall, that is, $h_i \leq h_j$ whenever $j_d\leq i_d$. Given $h\in \mathbb{R}$, one such external field is
$\widehat{\bm{h}}=(h_i)_{i\in\H+}$ with 
\begin{equation}
    h_i=\frac{h}{i_d^\delta}
\end{equation}

for all $i\in\H+$.  Figures \ref{EF.on.Lambda} and \ref{graph.EF} shows how this external field behaves. 
   \begin{figure}[ht]
            \centering
            \includegraphics[scale=0.15]{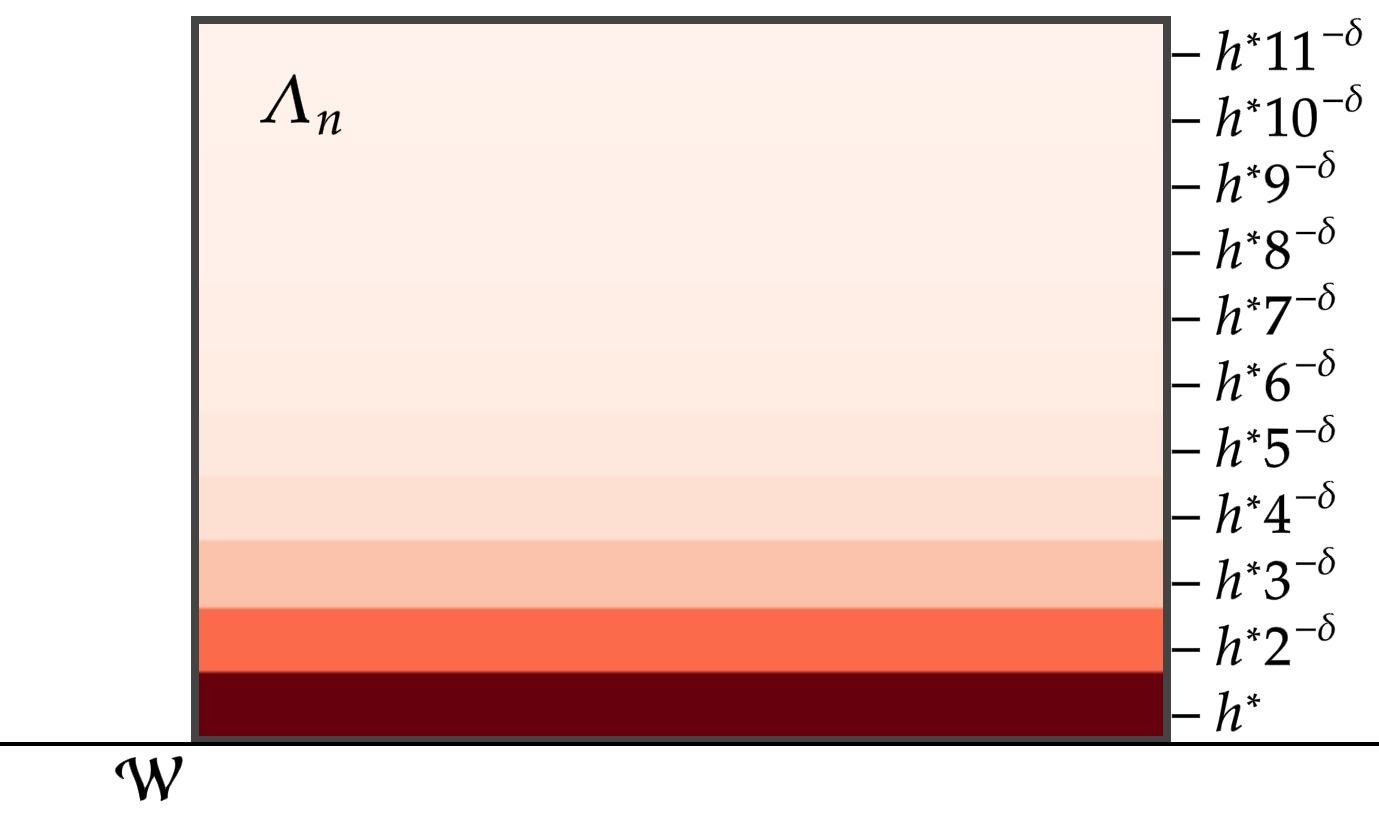}
            \caption{The influence of the external field in a box $\Lambda_n$.}
            \label{EF.on.Lambda}
    \end{figure}  
    \begin{figure}[ht]
            \centering
            \includegraphics[scale=0.15]{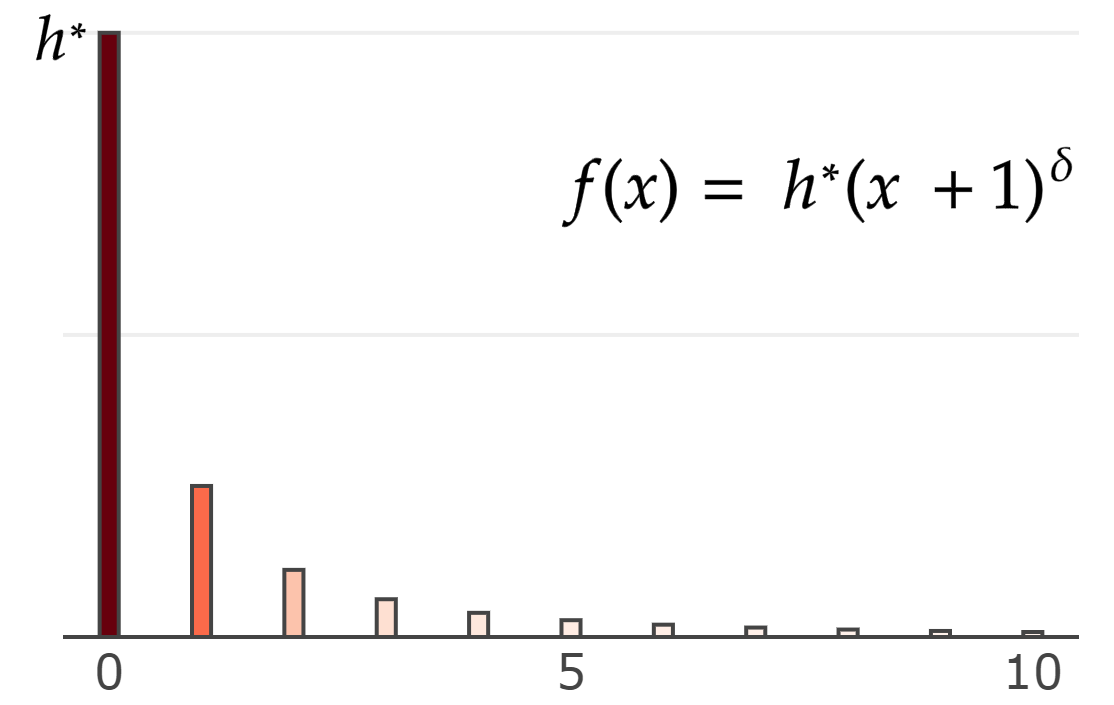}
            \caption{The external field w.r.t. the distance of a spin to the wall.}
            \label{graph.EF}
    \end{figure}
A particularly interesting choice of $h$ is $h=\lambda$, so we link the wall influence and the field. This particular case will be denoted  $\widehat{\bm{\lambda}}=(\lambda_i)_i\in\H+$, with 
\begin{equation}
    \lambda_i=\frac{\lambda}{i_d^\delta}.
\end{equation}
We can prove that, when $\delta>1$, the semi-infinite model with external field $\lambdadec$ behaves as the model with no field, so if we fix $\beta =1$, there exists a critical value $\overline{\lambda}_c(J)$ such that there are multiple Gibbs states when $0\leq\lambda<\overline{\lambda}_c(J)$, and we have uniqueness when $\overline{\lambda}_c(J)<\lambda$. At last, we show that when $\delta<1$, the semi-infinity Ising model with this choice of external field presents only one Gibbs state for any $J>0$. To simplify the notation, we choose to fix $\beta=1$ and let $J$ vary, as in the previous papers about the semi-infinite Ising model. Our results are summarized in the following theorem. 

\begin{theorem}
    Let $d\geq 2$, and let $J_c$ be the critical value of the Ising model in $\Z^d$ at $\beta=1$. Given any $\delta>0$, there exists a critical value $\overline{\lambda}_c=\overline{\lambda}_c(J,\delta)\geq 0$ such that the semi-infinite Ising model with external field $\lambdadec$ presents phase transition for all $0\leq \lambda < \overline{\lambda}_c$ and uniqueness for $\overline{\lambda}_c <\lambda$. Moreover, for $\delta>1$ and $J>J_c$, $0<\overline{\lambda}_c$. When $\delta<1$, $\overline{\lambda}_c=0$ and there is uniqueness for all $J>0$. 
\end{theorem}

The paper is organized as follows: in section 2 we introduce the model, some important definitions, and standard results. In section 3, we study the case $\delta>1$, showing the existence of the critical parameter $\overline{\lambda}_c>0$. In section 4, we prove that, for $\delta<1$, we have uniqueness. In the last section, we make a list of open problems for Ising-type models with fields.

\section{Preliminaries}   We replace  $\Z^d$ by the semi-lattice where we have a wall $\mathbb{H}^d_+ \coloneqq \{ i=(i_1, i_2, \dots, i_d)\in \mathbb{Z}^d : i_d\geq 1 \}.$

The {\it wall} is the subset $\mathcal{W}\coloneqq \{ i\in \mathbb{H}^d_+ : i_d = 1\}$. To represent the influence of the wall in its neighbors, we introduce a parameter $\lambda$ and, for the finite sets $\Lambda \Subset \mathbb{H}^d_+$, we define the Hamiltonian

\begin{equation}\label{SI.Hamiltonian}
    \mathcal{H}_{\Lambda; \lambda, \bm{h}}^{\bm{J}}(\sigma) \coloneqq -\sum_{\substack{i \sim j \\ \{i,j \} \cap \Lambda \neq \emptyset}} J_{i,j}\sigma_i\sigma_j - \sum_{i\in\Lambda}h_i\sigma_i - \sum_{i\in \Lambda\cap \mathcal{W}} \lambda\sigma_i.
\end{equation}
Here, ${\bm{h}=(h_i)_{i\in \mathbb{H}^d_+}}$ is the \textit{external field} and the interaction ${\bm{J}=(J_{i,j})_{i,j\in\mathbb{H}^d_+}}$ is non-negative for all $i,j\in\mathbb{Z}^d$, so we say the model is \textit{ferromagnetic}. Moreover, $i\sim j$ denotes $|i-j|=1$ with the norm being the $\ell_1$-norm. The local configurations in $\Lambda\Subset\mathbb{Z}^d$ with boundary condition $\eta\in\Omega$ are the elements of $\Omega_\Lambda^\eta \coloneqq \{ \omega\in \Omega : \omega_i=\eta_i \text{ for all } i \notin \Lambda \}$. The \textit{finite Gibbs measure in $\Lambda$ with $\eta$-boundary condition} is given by
\begin{equation}\label{eq:def.local.gibbs.measure}
    \mu_{\Lambda; \lambda, \bm{h}}^\eta (\sigma) \coloneqq \mathbbm{1}_{\{  \sigma\in\Omega_{\Lambda}^\eta\}}\frac{e^{-\mathcal{H}_{\Lambda; \lambda, \bm{h}}^{\bm{J}}(\sigma)}}{\mathcal{Z}^{\eta, \bm{J}}_{\Lambda; \lambda,\bm{h}}},
\end{equation}
 where $\mathcal{Z}^{\eta, \bm{J}}_{\Lambda; \lambda,\bm{h}}\coloneqq \sum_{\sigma \in\Omega_\Lambda^\eta}e^{-\mathcal{H}_{\Lambda; \lambda, \bm{h}}^{\bm{J}}(\sigma) }$ is the usual \textit{partition function}. This measure is defined over the $\sigma$-algebra generated by the cylinder sets, which coincides with the Borel $\sigma$-algebra when considering the product topology on $\Omega$, a compact space. Then, the set of probability measures defined over the Borel sets is a weak* compact set. 
 
To construct the infinite measures we consider sequences of finite subsets $(\Lambda_n)_{n\in\mathbb{N}}$  such that, for any subset $\Lambda\subset\H+$, there exists $N=N(\Lambda)>0$ such that $\Lambda\subset\Lambda_n$ for every $n>N$. We say such sequences \textit{invades} $\mathbb{H}^d_+$ and we denote it by $\Lambda_n\nearrow\H+$. A particularly important sequence that invades $\H+$ is the finite boxes 
$$\Lambda_{n, m} \coloneqq \{i\in \mathbb{H}^d_+ : i_d\leq m, -n \leq i_k \leq n \text{ for } k=1,\dots,d-1 \},$$ with $n,m\geq 0$. We define also $\mathcal{W}_{n,m} \coloneqq \mathcal{W}\cap \Lambda_{n, m}$ the restriction of the wall for these boxes. The set of \textit{Gibbs measures} $\mathcal{G}_{\bm{J}}$ is the closed convex hull of all the weak* limits obtained by sequences invading $\H+$:
\begin{equation}
\mathcal{G}_{\bm{J}} \coloneqq \overline{\text{conv}}\{\mu: \mu = w^*\text{-}\lim_{\Lambda^\prime \nearrow \H+}\mu_{\Lambda^\prime; \lambda, \bm{h}}^{\omega}\}.
\end{equation}
To simplify the notation, we are omitting the dependency on $\bm{h}$ and $\lambda$ in the definition of $\mathcal{G}_{\bm{J}}$. When $|\mathcal{G}_{\bm{J}}|>1$, we say that there is \textit{phase transition}, and when $|\mathcal{G}_{\bm{J}}|=1$, we have \textit{uniqueness}. 

The Hamiltonian of the Ising model in $\Lambda\Subset \mathbb{Z}^d$ is given by
\begin{equation}\label{Hamiltonian.Ising}
    \mathcal{H}_{\Lambda; \bm{h}}^{\bm{J}}(\sigma) = -\sum_{\substack{i \sim j \\ \{i,j \} \cap \Lambda \neq \emptyset}} J_{i,j}\sigma_i\sigma_j - \sum_{i\in\Lambda}h_i\sigma_i
\end{equation}
where ${\bm{J}=(J_{i,j})_{i,j\in\mathbb{Z}^d}}$ a family of non-negative real number and  ${\bm{h}=(h_i)_{i\in \mathbb{Z}^d}}$ is the external field with $h_i\in\mathbb{R}$ for all $i\in\mathbb{Z}^d$.

Replacing the semi-infinite lattice $\H+$ by the whole lattice $\Z^d$ in the definitions above, the set of \textit{Ising Gibbs measures} with ferromagnetic interaction $\bm{J}=(J_{i,j})_{i,j\in \Z^d}$ and external field $\bm{h}=(h_i)_{i\in\Z^d}$ is
\begin{equation*}
    \mathcal{G}_{\bm{J}}^{IS} \coloneqq \overline{\text{conv}}\{\mu: \mu = w^*\text{-}\lim_{\Lambda^\prime \nearrow \Z^d}\mu_{\Lambda^\prime; 0, \bm{h}}^{\omega}\}.
\end{equation*}


Similarly, we can define finite volume Gibbs measures for the Ising model using the Hamiltonian \eqref{Hamiltonian.Ising}. The \textit{finite Gibbs measure in $\Lambda$ with $\eta$-boundary condition for the Ising model} is completely determined by the integral of the local functions $f$, 
\begin{equation*}
    \langle f \rangle_{\Lambda; \bm{h}}^\eta \coloneqq \mu_{\Lambda; 0, \bm{h}}^\eta (f) = (Q^{\eta, \bm{J}}_{\Lambda; \bm{h}})^{-1} \sum_{\sigma\in\Omega_\Lambda^\eta} f(\sigma)e^{-\mathcal{H}_{\Lambda; \bm{h}}^{\bm{J}}(\sigma)}, 
\end{equation*}
where 
\begin{equation*}
    Q^{\eta, J}_{\Lambda; \bm{h}} = \sum_{\sigma\in \Omega^{\eta}_{\Lambda}}\exp{ \{ \sum_{\substack{i \sim j \\ \{i,j \} \cap \Lambda \neq\emptyset}} J\sigma_i\sigma_j + \sum_{i\in\Lambda} h_i\sigma_i \} }
\end{equation*}
is the usual partition function. 

The semi-infinite model inherits several properties from the usual Ising model in $\mathbb{Z}^d$, since, 
for $\Lambda\Subset\mathbb{H}^d_+$, the state $\langle f\rangle^{\eta}_{\Lambda; \lambda, \bm{h}}$ is the Ising state with interaction $(J^\lambda_{i,j})_{i,j\in\mathbb{Z}^d}$ given by
\begin{equation}\label{J.for.the.semi.infinite}
J^\lambda_{i,j}=\begin{cases} J_{i,j} &\text{ if } \{i,j\}\subset\mathbb{H}^d_+, \\ \lambda &\text{ otherwise,} \end{cases}
\end{equation}
 and boundary condition $\eta^+$ given by, for all $i\in\mathbb{Z}^d$,
 \begin{equation*}\label{eta^+}
     \eta^+_i = \begin{cases} \eta &\text{ if } i\in\mathbb{H}^d_+, \\ +1 &\text{ otherwise.} \end{cases}
 \end{equation*}

We are interested in phase transition results for a uniform interaction $\bm{J} \equiv J>0$, $\lambda>0$ and $h_i\geq 0$ for all $i\in\H+$. For the Ising model with no external field, the existence of two distinct states translates to the fact that, on a macroscopic scale, the spins will align in the same direction. If we consider the plus boundary condition, when we look at boxes, we will see spins at the plus phase (positive mean) at low temperatures. The same occurs for the minus boundary condition.
        
        In the semi-infinite Ising model, again with no field, the macroscopic consequence of the phase transition is different, it has to do with the existence of a layered phase separating the wall from the bulk. Writing the semi-infinite model interaction as in (\ref{J.for.the.semi.infinite}), the minus state in a box is given by the boundary condition like 
       in  Figure \ref{minus.b.c},
\begin{figure}[ht]
\centering
\includegraphics[scale=0.5]{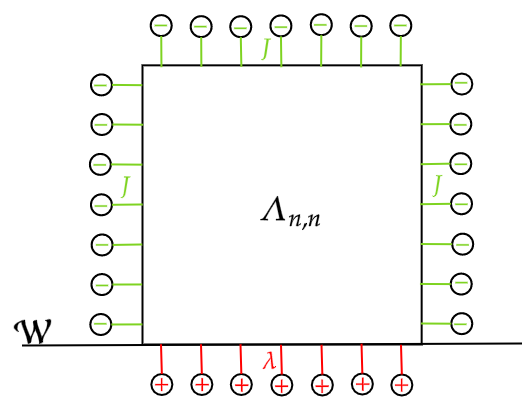}
\caption{The $\textcolor{red}{+}$ spins in the wall compete with the $\textcolor{green}{-}$ on the boundary.}
\label{minus.b.c}
\end{figure}
        so if $\lambda$ is big enough we have the phenomenon of \textit{Complete wetting}, where the wall forces to spin to align in the plus direction. This macroscopic behavior is shown in \cite{FP-II} and an improvement for low temperatures can be found on \cite{Holicky_Zahradnik_93}. To better understand this type of phenomenon, we recommend the survey paper \cite{IV18}. The surface free energy of the wall is a quantity that tries to identify whether or not we have complete wetting. 

Consider the sequence that invades $\H+$ given by $\Lambda_n \coloneqq \Lambda_{n,n}$. Take  $\Lambda^\prime_n$ as the reflection of $\Lambda_n$ with respect to the line ${\mathcal{L} \coloneqq \{ (i_1,\dots,i_d)\in \mathbb{Z}^d : i_d = \frac{1}{2}\}}$. Similarly, define the reflection of the walls $\mathcal{W}_n \coloneqq [-n,n]^{d-1}\times \{1\}$ as $\mathcal{W}^\prime_n \coloneqq [-n,n]^{d-1}\times\{0\}$ and denote $\Delta_n \coloneqq \Lambda_n \cup \Lambda^\prime_n$ the extended box. For any summable sequence of positive real numbers $\mathrm{h}=(\mathrm{h}_\ell)_{\ell\geq 1}$, let $\bm{\mathrm{h}} = \{\mathrm{h}_{i}\}_{i\in\H+^d}$ be the external field induced by $\mathrm{h}$, that is, $\mathrm{h}_i\coloneqq\mathrm{h}_{i_d}$, with $i_d$ being the last coordinate of $i\in\Z^d$. We also denote $\overline{\bm{\rmh}} = \{\overline{\rmh}_i\}_{i\in\Z^d}$ the natural extension of $\bm{\rmh}$ to $\Z^d$, defined by
\begin{equation}\label{Def: overline_h}
    \overline{\rmh}_i = \begin{cases}
                            \rmh_i & \text{ if }i\in\H+ \\
                            \rmh_{-i + e_d} & \text{ if }i\in\Z^d\setminus\H+,
                        \end{cases}
\end{equation}
where $e_d=(0,\dots,0,1)$ is a canonical base vector. Given any $J>0$ and summable  $\mathrm{h}=(\mathrm{h}_\ell)_{\ell\geq 0}$, the \textit{surface free energy} for the $+$-boundary condition and $-$-boundary condition are, respectively, 

\begin{equation}\label{Def.F+}
    F^{+}( J, \lambda, \bm{\rmh}) \coloneqq \lim_{n\to \infty} -\frac{1}{2|\mathcal{W}_n|}\ln\left[ \frac{(\mathcal{Z}^{+, J}_{n; \lambda,\bm{\rmh}})^2}{Q^{+, J}_{\Delta_n; 0}} \right] 
\end{equation}
and
\begin{equation}\label{Def.F-}
    F^{-}( J, \lambda, \bm{\rmh}) = \lim_{n\to \infty} -\frac{1}{2|\mathcal{W}_n|}\ln\left[ \frac{(\mathcal{Z}^{-, J}_{n; \lambda,\bm{\rmh}})^2}{Q^{-, J}_{\Delta_n; 0}} \right].
\end{equation}

One of our first goals is to prove these limits exist. After that, we will use the \textit{wall free energy} between the $+$ and $-$ b.c., defined as 
\begin{equation}\label{Def: Wall free energy}
\tau_w(J, \lambda, \bm{\rmh}) \coloneqq F^{-}(J, \lambda , \bm{\rmh}) - F^{+}(J, \lambda, \bm{\rmh})
\end{equation}
to characterize the presence or absence of phase transition.

Our proof follows \cite{FP-I} closely, with the key difference being the definition of the surface free energy. In \cite{FP-I}, both partition functions in the definition of $F^{\pm}(J,\lambda, h)$ have the same external field $h$, so the surface free energy measures only the influence of $\lambda$ in the wall. With this difference, our definition of surface tension with no field $\tau_w(J,\lambda, 0)$ matches their definition, but the same does not hold when we add an external field. However, our definition has a downside: the limits \eqref{Def.F+} and $\eqref{Def.F-}$ only exist for external field $\bm{\rmh}$ induced by summable external fields. As the parameter $\lambda$ can be incorporated in the external field $\bm{\rmh}$ by adding $\lambda$ to $\rmh_1$, we often omit $\lambda$ in the notation.

In our proofs, we use two consequences of duplicated variables inequalities, the next two Propositions. The proof follows the same steps as the one in \cite{FP-I}, so we omit it.
\begin{proposition}\label{basta.comparar.magnetizacoes}
Let $\bm{J}=(J_{i,j})_{i,j\in\mathbb{Z}^d}$ be a non-negative interaction satisfying $J_{i,j}>0$ if $|i-j|=1$, $\bm{h}=(h_i)_{i\in\mathbb{Z}^d}$ be a non-negative external field. Then, if $\langle \sigma_i \rangle^{-}_{\bm{h}} = \langle \sigma_i \rangle^{+}_{\bm{h}}$ for some $i\in\mathbb{Z}^d$, then there exists a unique Gibbs state.  
\end{proposition}

\begin{proposition}\label{Consequence.of.DVI}
Let $\lambda\geq0$ and $\bm{h}=(h_{i})_{i\in\mathbb{H}_d^+}$ be a family of non-negative real numbers. Then, for all $\Lambda\Subset \mathbb{H}_d^+$ and all $i,j\in\mathbb{H}_d^+$
\begin{equation}
    \langle \sigma_i\sigma_j \rangle_{\Lambda; \lambda,\bm{h}}^-\leq  \langle \sigma_i\sigma_j \rangle_{\Lambda; \lambda,\bm{h}}^+ \label{increasing.correlations}
\end{equation}
and
\begin{equation}
\langle \sigma_i\sigma_j \rangle_{\Lambda; \lambda,\bm{h}}^+ -  \langle \sigma_i \rangle_{\Lambda; \lambda,\bm{h}}^+ \langle \sigma_j \rangle_{\Lambda; \lambda,\bm{h}}^+ \leq   \langle \sigma_i\sigma_j \rangle_{\Lambda; \lambda,\bm{h}}^- -  \langle \sigma_i \rangle_{\Lambda; \lambda,\bm{h}}^- \langle \sigma_j \rangle_{\Lambda; \lambda,\bm{h}}^-
\end{equation}
\end{proposition}

\section{Critical behaviour when \texorpdfstring{$\delta>1$}{d>1}}In the same steps as in \cite{FP-I}, we prove that the surface free energies \eqref{Def.F+} and \eqref{Def.F-} are well defined. Then, we follow \cite{FP-II} to prove the existence of a critical value $\overline{\lambda}_c$ such that, for $\lambda>\overline{\lambda}_c$ there is a unique state, and for $\lambda<\overline{\lambda}_c$, there is a phase transition.

\begin{proposition}
    For any $J\geq 0$ and any summable sequence of positive real numbers $\mathrm{h}=(\mathrm{h}_\ell)_{\ell\geq 1}$, The limits $ {F}^+(J,\bm{\mathrm{h}})$ and $ {F}^-(J,\bm{\mathrm{h}})$ are well defined. Here, $\bm{\rmh}$ is the external field induced by $\mathrm{h}$.
\end{proposition}
\begin{proof}
    As the parameters $J$ and $\bm{\rmh}$ are fixed, we will omit them from the notation. Also, in the sums, we omit $|i- j|=1$, since this is always the case. With the definition of $\overline{\bm{\rmh}}$ in \eqref{Def: overline_h} we have
    \begin{align*}
        \left(\mathcal{Z}^+_{n; \bm{\rmh}}\right)^2 &= \sum_{\sigma\in \Omega^+_{\Delta_n}}\exp{\left\{ \sum_{i,j\in\Delta_n}J\sigma_i\sigma_j - \sum_{i\in\calW_n} J\sigma_i\sigma_{i-e_d} + \sum_{i\in \Delta_n}\overline{\rmh}_i\sigma_i + \sum_{\substack{i\in \Delta_n \\ j\notin \Delta_n}}J\sigma_i \right\}} \\ 
        & = \sum_{\sigma\in \Omega^+_{\Delta_n}}\exp{\left\{ \sum_{i,j\in\Delta_n}J\sigma_i\sigma_j + \sum_{\substack{i\in \Delta_n \\ j\notin \Delta_n}}J\sigma_i - \Tilde{H}_n(\sigma)\right\}}, 
    \end{align*}
    with $\Tilde{H}_n(\sigma) = \sum_{i\in\calW_n} J\sigma_i\sigma_{i-e_d} - \sum_{i\in\Delta_n}\overline{\rmh}_i\sigma_i$. Take 
    \begin{equation*}
         \Xi_n(t) \coloneqq  \sum_{\sigma\in \Omega^+_{\Delta_n}}\exp{\left\{ \sum_{i,j\in\Delta_n}J\sigma_i\sigma_j + \sum_{\substack{i\in \Delta_n \\ j\notin \Delta_n}}J\sigma_i - t\Tilde{H}_n(\sigma)\right\}},
    \end{equation*}
    and let $\langle\cdot \rangle^{+}_{\Delta_n}(t)$ be the state given by the Hamiltonian $H_n(t)(\sigma) = \sum_{i,j\in\Delta_n}J\sigma_i\sigma_j + \sum_{\substack{i\in \Delta_n \\ j\notin \Delta_n}}J\sigma_i - t\Tilde{H}_n(\sigma)$ with $+$-boundary condition. We can write 
    \begin{align}\label{Eq: F_as_integral}
        \ln{\left[\frac{\left(\mathcal{Z}_{n;\bm{\rmh}}^+\right)^2}{Q_{\Delta_n; J}^+}\right]} &= \ln{\left[\frac{\Xi_n(1)}{\Xi_n(0)}\right]} = \int_{0}^1 \frac{\d}{\d s}\left(\ln\Xi_n(s)\right)ds \nonumber \\
        &= - \int_0^1\sum_{i\in\calW_n} J \langle \sigma_i\sigma_{i-e_d} \rangle_{\Delta_n}^+(s) ds + \int_0^1 \sum_{i\in\Delta_n}\overline{\rmh_i}\langle\sigma_i \rangle_{\Delta_n}^+(s)ds.
    \end{align}
    Considering $\langle \cdot \rangle^+(t)\coloneqq\lim_{n\to\infty} \langle \cdot \rangle^+_{\Delta_n}(t)$ the limiting state, by standard arguments we can use that the plus states are decreasing in $\Lambda$, with the order on the inclusion, to show that 
       \begin{equation}\label{Eq: Lim_h_i_sigma_i}
        \lim_{n\to\infty} \frac{1}{|\calW_n|}\sum_{i\in\Delta_n}  \overline{\rmh}_i\langle \sigma_{i} \rangle_{\Delta_n}^+(s) = \sum_{\ell=1}^{\infty} \rmh_\ell \left(\langle \sigma_{\ell e_d} \rangle^+(s) + \langle \sigma_{-(\ell -1)e_d} \rangle^+(s)\right),
    \end{equation}

    and 

    \begin{equation}\label{Eq: Lim_correlation_on_wall}
        \lim_{n\to\infty} \frac{1}{|\calW_n|}\sum_{i\in\calW_n}  \langle \sigma_i\sigma_{i-e_d} \rangle_{\Delta_n}^+(s) = \langle \sigma_0\sigma_{-e_d} \rangle^+(s).
    \end{equation}
    As the monotonicity only holds for increasing functions, one can use the transformation $\eta_i = \frac{(\sigma_i + 1)}{2}$  for all $i\in\Z^d$, so $\eta_i\eta_j$ and $\eta_i$ are increasing. Equations \eqref{Eq: F_as_integral}, \eqref{Eq: Lim_correlation_on_wall} and \eqref{Eq: Lim_h_i_sigma_i}, together with the dominated convergence theorem, yields
\begin{align*}
    F^+(J,\bm{\rmh}) &= \lim_{n\to\infty}\frac{1}{|\calW_n|} \int_0^1\sum_{i\in\calW_n} J \langle \sigma_i\sigma_{i-e_d} \rangle_{\Delta_n}^+(s) ds -\frac{1}{|\calW_n|}  \int_0^1 \sum_{i\in\Delta_n}\overline{\rmh_i}\langle\sigma_i \rangle_{\Delta_n}^+(s)ds \\
    &= \int_0^1 J\langle \sigma_0\sigma_{-e_d} \rangle^+(s) ds - \sum_{\ell=1}^{\infty} \int_0^1  \rmh_\ell \left(\langle \sigma_{\ell e_d} \rangle^+(s) + \langle \sigma_{-(\ell -1)e_d} \rangle^+(s)\right) ds,
\end{align*}
so ${F}^+(J,\bm{\rmh})$ is well defined. The proof that the limit ${F}^-(J,\bm{\rmh})$ exists is analogous. 
 
\end{proof} 
 
To characterize the phase transition, we proceed as in \cite{FP-II} and use the {wall free energy}, defined in \eqref{Def: Wall free energy}. Notice that, when we do not have an external field, $Q_{\Delta_n}^+ = Q_{\Delta_n}^-$. This simplifies the wall free energy to 
\begin{equation}\label{tau_tilde_w_wo_ising_partition}
    {\tau}_w(J,\bm{\rmh}) = \lim_{n \to \infty} -\frac{1}{|\mathcal{W}_n|}\ln\left[ \frac{\mathcal{Z}^{-}_{n; \bm{\rmh}}}{\mathcal{Z}^{+}_{n; \bm{\rmh}}} \right].
\end{equation}
First, we prove that, for the external field $\widehat{\bm{\lambda}}$, we can write $ {\tau}_w(J,\widehat{\bm{\lambda}})$ in terms of differences of the magnetization. 

\begin{proposition}\label{Prop: Tilde_tau_as_integral}
For $J>0$ and $\lambda\geq0$, the wall free energy can be written as 
\begin{equation}
     {\tau}_w(J,\widehat{\bm{\lambda}}) = \int_0^\lambda \sum_{\ell=1}^\infty \frac{1}{\ell^{\delta}}\left( \langle \sigma_{\ell e_d} \rangle^+_{J, \widehat{\bm{s}}} - \langle \sigma_{\ell e_d} \rangle^-_{J, \widehat{\bm{s}}} \right) ds.
\end{equation}
\end{proposition}

\begin{proof}
    With the simplification  \eqref{tau_tilde_w_wo_ising_partition} in mind, we differentiate each term w.r.t. $\lambda$ to get
\begin{align*}
    -\partial_\lambda \left( \ln\left[ \frac{\mathcal{Z}_{n;\lambdadec}^{-}}{\mathcal{Z}_{n;\lambdadec}^{+}} \right] \right) 
        &= \frac{1}{\mathcal{Z}_{n;\lambdadec}^{+}} \partial_\lambda \left( \mathcal{Z}_{n;\lambdadec}^{+} \right) - \frac{1}{\mathcal{Z}_{n;\lambda}^{-}} \partial_\lambda \left( \mathcal{Z}_{n;\lambdadec}^{-} \right).
\end{align*}
As 
\begin{equation*}
    \partial_\lambda\left(\mathcal{Z}_{n;\lambdadec}^{+}\right)=\sum_{i\in \Lambda_n}\sum_{\sigma\in\Sigma_{\Lambda_n}^{+}} \frac{1}{i_d^\delta}\sigma_i e^{-\mathcal{H}_{\Lambda_n; \lambdadec}^J(\sigma)} = \sum_{i\in \calW_n^\prime}\sum_{\ell=1}^{n} \frac{1}{\ell^\delta}\sigma_{i+\ell e_d} e^{-\mathcal{H}_{\Lambda_n; \lambdadec}^J(\sigma)} ,
\end{equation*}
we conclude that
\begin{equation}\label{d.lambda.of.ln[Z+/Z-].decaying.field}
     -\partial_\lambda \left( \ln\left[ \frac{\mathcal{Z}_{n;\lambda}^{-}}{\mathcal{Z}_{n;\lambda}^{+}} \right] \right) = \sum_{i\in \mathcal{W}_n^\prime} \sum_{\ell=1}^{n} \frac{1}{\ell^\delta}\left(\langle \sigma_{i+\ell e_d} \rangle^{+}_{n;\lambdadec} - \langle \sigma_{i+\ell e_d} \rangle^{-}_{n;\lambdadec}\right).
\end{equation}
All of the above functions are continuous and bounded since they are the logarithm of positive polynomials. Moreover, $\mathcal{Z}^{+}_{n; 0} = \mathcal{Z}^{-}_{n; 0}$. Hence we can write
\begin{equation*}
     {\tau}_w(J,\lambdadec) = \lim_{n\to \infty} \frac{1}{|\mathcal{W}_n|}\sum_{i\in \mathcal{W}_n^\prime}\sum_{\ell=1}^{n}\int_0^\lambda \frac{1}{\ell^\delta}\left(\langle \sigma_{i+\ell e_d} \rangle^{+}_{n;\widehat{\bm{s}}} - \langle \sigma_{i+\ell e_d} \rangle^{-}_{n;\widehat{\bm{s}}}\right)ds.
\end{equation*}
The result follows from the dominated convergence theorem once we note that, 
\begin{equation*}
    \lim_{n\to \infty} \frac{1}{|\mathcal{W}_n|}\sum_{i\in \mathcal{W}_n^\prime}\sum_{\ell=1}^{n} \frac{1}{\ell^\delta}\langle \sigma_{i+\ell e_d} \rangle^{+}_{n;\widehat{\bm{s}}} = \sum_{\ell=1}^\infty  \frac{1}{\ell^\delta}\langle \sigma_{\ell e_d} \rangle^{+}_{\widehat{\bm{s}}} 
\end{equation*}
and 
\begin{equation*}
    \lim_{n\to \infty} \frac{1}{|\mathcal{W}_n|}\sum_{i\in \mathcal{W}_n^\prime}\sum_{\ell=1}^{n} \frac{1}{\ell^\delta}\langle \sigma_{i+\ell e_d} \rangle^{-}_{n;\widehat{\bm{s}}}= \sum_{\ell=1}^\infty  \frac{1}{\ell^\delta}\langle \sigma_{\ell e_d} \rangle^{-}_{\widehat{\bm{s}}}. 
\end{equation*}
Again, the proof of these limits follows by standard arguments using the monotonicity of the states with respect to $\Lambda$. 
\end{proof}

This new wall-free energy also presents the monotonicity and convexity properties of the previous one. Such properties are described in the next proposition. 

\begin{proposition}\label{Prop: Monotonicity_and_conv_of_tilde_tau}
    For $J>0$, and an external field $\bm{\rmh}$ induced by a positive, summable sequence $\rmh = (\rmh_\ell)_{\ell=1}^\infty$, and $\lambda>0$, we have
    \begin{itemize}
        \item[(a)] ${\tau}_w(J,\bm{\rmh})$ is non-decreasing in $J$ and $\rmh_\ell$, for all $\ell\geq 1$; 
        \item[(b)]${\tau}_w(J, \lambdadec)$ is a concave function of $\lambda >0$. 
    \end{itemize}
\end{proposition}

\begin{proof}
    Item \textit{(a)} follows from the representation \eqref{tau_tilde_w_wo_ising_partition} after we differentiate the limit term with respect to the appropriate variable. Differentiating the term in the limit with respect to $J$ we get
\begin{equation*}
    -\partial_J\left(\frac{1}{|\mathcal{W}_n|}\ln\left[ \frac{\mathcal{Z}^{-}_{n;\bm{\rmh}}}{\mathcal{Z}^{+}_{n;\bm{\rmh}}} \right] \right) = |\mathcal{W}_n|^{-1} \sum_{\substack{i\sim j \\ \{ i,j \} \cap \Lambda_n \neq \emptyset}}   \langle \sigma_i\sigma_j \rangle_{n; \bm{\rmh}}^+ -  \langle \sigma_i\sigma_j \rangle_{n; \bm{\rmh}}^-,
\end{equation*}
that is positive by Proposition \ref{Consequence.of.DVI}, a consequence of the duplicate variables inequalities. Differentiating the same term we respect to $\rmh_\ell$ for a fixed $\ell\geq 1$ we have
\begin{equation*}
    -\partial_{\rmh_\ell}\left(\frac{1}{|\mathcal{W}_n|}\ln\left[ \frac{\mathcal{Z}^{-}_{n; \bm{\rmh}}}{\mathcal{Z}^{+}_{n; \bm{\rmh}}} \right] \right) = |\mathcal{W}_n|^{-1} \sum_{i\in\calW_n^\prime}   \langle \sigma_{i+\ell e_d} \rangle_{n; \bm{\rmh}}^+ -  \langle \sigma_{i+ \ell e_d} \rangle_{n; \bm{\rmh}}^-,
\end{equation*}
that is positive by FKG. 
To prove claim \textit{(b)}, we use a similar reasoning. By equation \eqref{d.lambda.of.ln[Z+/Z-].decaying.field}, we have 
\begin{multline*}
    -\partial^2_\lambda \left( \ln\left[ \frac{\mathcal{Z}_{n;\lambdadec}^{-}}{\mathcal{Z}_{n;\lambdadec}^{+}} \right] \right) = \sum_{i\in \mathcal{W}_n^\prime} \sum_{\ell=1}^{n} \frac{1}{\ell^\delta}\partial_{\lambda}\left(\langle \sigma_{i+\ell e_d} \rangle^{+}_{n;\lambdadec} - \langle \sigma_{i+\ell e_d} \rangle^{-}_{n;\lambdadec}\right) \\
    = \sum_{i,j\in \mathcal{W}_n^\prime} \sum_{\ell,k=1}^{n} \frac{1}{\ell ^\delta}\frac{1}{k^\delta}\left(\langle \sigma_{i+\ell e_d} \sigma_{j+k e_d}\rangle^{+}_{n;\lambdadec} - \langle \sigma_{i+\ell e_d} \rangle^{+}_{n;\lambdadec}\langle \sigma_{j+k e_d}\rangle^{+}_{n;\lambdadec}   - \langle \sigma_{i+\ell e_d} \rangle^{-}_{n;\lambdadec} + \langle \sigma_{i+\ell e_d} \rangle^{-}_{n;\lambdadec}\langle \sigma_{j+k e_d}\rangle^{-}_{n;\lambdadec} \right),
\end{multline*}
that is smaller or equal to zero by Proposition \ref{Consequence.of.DVI}. So ${\tau}_w$ is the limit of concave functions, and therefore it is concave.
\end{proof}

To relate the wall-free energy and the phase transition or uniqueness, we introduce the critical quantity
\begin{equation*}
    \overline{\lambda}_{c}(J)\coloneqq \inf\{\lambda \geq 0: {\tau}_w(J,\lambdadec) = \max_{s\geq 0}{\tau}_w(J,\widehat{\bm{s}})\}.
\end{equation*}
Using Proposition \ref{Prop: Monotonicity_and_conv_of_tilde_tau}, in the next Lemma we show that the wall free energy reaches a maximum and therefore $\overline{\lambda}_c$ is finite. We do so by comparing it to    
\begin{equation}\label{lambda_c.by.states}
    \lambda_c \coloneqq \inf \{\lambda\geq 0 : \langle \sigma_0 \rangle^{+}_{\lambda,0} = \langle \sigma_0 \rangle^{-}_{\lambda,0} \},
\end{equation}
the critical value of the semi-infinite model with no external field. In \cite{FP-II}, it is shown that, $\tau_w(J, \lambda, 0) = \tau(J)$ for all $\lambda>\lambda_c$, where $\tau(J)$ is the \textit{interface free energy for the Ising model}, defined as
\begin{equation}\label{tau(J)}
    \tau(J) = \lim_{n,m\to\infty} -\frac{1}{|\mathcal{W}_n|} \ln\left[ \frac{Q^{\mp, J}_{\Delta_{m,n}; 0}}{Q^{+, J}_{\Delta_{m,n}; 0}} \right].
\end{equation}
The $\mp$-boundary condition denotes the configuration $\sigma\in\Omega$ defined by 
\begin{equation*}
    \sigma_i=\begin{cases}
                -1, & \text{ if } i\in \mathbb{H}^d_+, \\
                +1 & \text{ if } i\in \Z^d\setminus \mathbb{H}_+^d
                \end{cases}
\end{equation*}
and 
\begin{equation*}
    \Delta_{m,n} = \left[-m,m\right]^{d-1}\times\left[-n,n\right].
\end{equation*}
    
\begin{lemma}\label{Lemma: Comparing_taus}
    For $J>0$, $\overline{\lambda}_c$ is finite and $\overline{\lambda}_c\leq\lambda_c$. Moreover, $\overline{\lambda}_c>0$ whenever  $J>J_c$. 
\end{lemma}
\begin{proof}
  To prove that $\overline{\lambda}_c\leq\lambda_c$, it is enough to show that for $J>0$ and $\lambda\geq \lambda_c$, ${\tau}_w(J,\lambdadec) = {\tau}_w(J,\widehat{\bm{\lambda_c}})$.

     Indeed, for all $i,j\in \H+$, $n\in\mathbb{N}$ and positive external field $\bm{h}$, $\langle \sigma_i\rangle_{n;\bm{h}}^+ - \langle \sigma_i\rangle_{n;\bm{h}}^-$ is decreasing in $h_j$, since 
    \begin{equation}\label{Eq: Diff_mag_is_decreasing}
       \partial_{h_j}(\langle \sigma_i\rangle_{n;\bm{h}}^+ - \langle \sigma_i\rangle_{n;\bm{h}}^-) = \langle \sigma_i\sigma_j\rangle_{n;\bm{h}}^+  -\langle \sigma_i\rangle_{n;\bm{h}}^+\langle \sigma_j\rangle_{n;\bm{h}}^+ - \langle \sigma_i\sigma_j\rangle_{n;\bm{h}}^- + \langle \sigma_i\rangle_{n;\bm{h}}^-\langle \sigma_j\rangle_{n;\bm{h}}^- \leq 0
    \end{equation}
    by Proposition \ref{Consequence.of.DVI}. In particular, for any $\lambda\geq 0$, $\langle \sigma_i\rangle_{n;\lambdadec}^+ - \langle \sigma_i\rangle_{n;\lambdadec}^- \leq \langle \sigma_i\rangle_{n;\lambda, 0}^+ - \langle \sigma_i\rangle_{n;\lambda, 0}^-$, and the same inequality holds for the limit states. As $\langle \sigma_i\rangle_{\lambda, 0}^+ - \langle \sigma_i\rangle_{\lambda, 0}^- = 0$ for all $\lambda > \lambda_c$ and $i\in\H+$, using Proposition \ref{Prop: Tilde_tau_as_integral} we conclude that 
    \begin{equation*}
        {\tau}_w(J,\lambdadec) = \int_0^\lambda \sum_{\ell=1}^\infty \frac{1}{\ell^{\delta}}\left( \langle \sigma_{\ell e_d} \rangle^+_{J, \widehat{\bm{s}}} - \langle \sigma_{\ell e_d} \rangle^-_{J, \widehat{\bm{s}}} \right) ds = \int_0^{\lambda_c} \sum_{\ell=1}^\infty \frac{1}{\ell^{\delta}}\left( \langle \sigma_{\ell e_d} \rangle^+_{J, \widehat{\bm{s}}} - \langle \sigma_{\ell e_d} \rangle^-_{J, \widehat{\bm{s}}} \right) ds = {\tau}_w(J,\widehat{\bm{\lambda_c}}).
    \end{equation*}
  By the monotonicity on the external field, given $\lambda\geq 0$ and taking $\bm{\lambda}_0 \coloneqq \{\lambda \mathbbm{1}_{\{i\in\calW\}}\}$, the external field that is zero outside of $\calW$, we have
     \begin{equation}\label{Eq: tau_smaller_than_tilde_tau}
       \tau_w(J,\lambda) = {\tau}_w(J,\bm{\lambda}_0)\leq {\tau}_w(J,\lambdadec)
    \end{equation}   
     For $J> J_c$, it was shown in \cite{Lebowitz_Pfister_81} that  $\tau(J)>0$. As $\tau(J)=\tau_w(J,\lambda)$ for all $\lambda\geq\lambda_c$, inequality \eqref{Eq: tau_smaller_than_tilde_tau} yields
\begin{equation*}
    0<\tau_w(J,\lambda_c)\leq {\tau}_w(J,\widehat{\bm{\lambda_c}}),
\end{equation*}
and therefore $\overline{\lambda}_c>0$ whenever $J>J_c$.
\end{proof}

We end this section by proving that $\overline{\lambda}_c$ is the critical value for phase transition. 
\begin{proposition}
    For any $0\leq \lambda<\overline{\lambda}_c$, $\langle \cdot \rangle_{J,\lambdadec}^+ \neq  \langle \cdot \rangle_{J,\lambdadec}^-$. And for $\lambda>\overline{\lambda}_c$,  $\langle \cdot \rangle_{J,\lambdadec}^+ =  \langle \cdot \rangle_{J,\lambdadec}^-$.
\end{proposition}
\begin{proof}
    Fixed $0\leq \lambda < \overline{\lambda}_c$, lets assume by contradiction that $\langle \cdot \rangle_{J,\lambdadec}^+ =  \langle \cdot \rangle_{J,\lambdadec}^-$. As we argued before, inequality \eqref{Eq: Diff_mag_is_decreasing} shows that the difference $\langle \sigma_i \rangle_{n;J,\lambdadec}^+ - \langle \sigma_i \rangle_{n;J,\lambdadec}^-$ is decreasing in $\lambda$ for all $i\in\H+$. Hence, for every $\lambda^\prime>\lambda$, 
    \begin{equation*}
        \langle \sigma_i \rangle_{n;J,\widehat{\bm{\lambda^\prime}}}^+ - \langle \sigma_i \rangle_{n;J,\widehat{\bm{\lambda^\prime}}}^- \leq \langle \sigma_i \rangle_{n;J,\lambdadec}^+ - \langle \sigma_i \rangle_{n;J,\lambdadec}^- = 0.
    \end{equation*}
By Proposition \ref{Prop: Tilde_tau_as_integral}, this implies that 
\begin{equation*}
     {\tau}_w(J,\lambdadec) = \int_0^\lambda \sum_{\ell=1}^\infty \frac{1}{\ell^{\delta}}\left( \langle \sigma_{\ell e_d} \rangle^+_{J, \widehat{\bm{s}}} - \langle \sigma_{\ell e_d} \rangle^-_{J, \widehat{\bm{s}}} \right) ds = \int_0^{\overline{\lambda}_c} \sum_{\ell=1}^\infty \frac{1}{\ell^{\delta}}\left( \langle \sigma_{\ell e_d} \rangle^+_{J, \widehat{\bm{s}}} - \langle \sigma_{\ell e_d} \rangle^-_{J, \widehat{\bm{s}}} \right) ds = \max_{s\geq 0}{\tau}_w(J,\widehat{\bm{s}}).
\end{equation*}
In the last equation, we are using that the maximum is reached at $\overline{\lambda}_c$ since all concave functions are continuous. This shows that $\lambda\geq \overline{\lambda}_c$, a contradiction. 

 For $\lambda>\overline{\lambda}_c$, since ${\tau}_w$ is non-decreasing in $\lambda$, ${\tau}_w(J,\lambdadec) = \max_{s\geq 0}{\tau}_w(J,\widehat{\bm{s}})$. Moreover, it is differentiable in $\lambda$ and is the point-wise limit of the sequence ${|\mathcal{W}_n|^{-1}}\left(\ln{\mathcal{Z}^{+}_{n; \lambdadec} - \ln {\mathcal{Z}^{-}_{n; \lambdadec}}}\right)$, which is concave. We can then use a known theorem for convex functions, for example, \cite[Theorem B.12 ]{FV-Book}, to conclude that
\begin{align*}
    0 = \partial_\lambda{\tau}_w(J,\lambdadec) = \lim_{n\to\infty} \frac{1}{\left|\calW_n\right|}\partial_\lambda\left(\ln{\mathcal{Z}^{+}_{n; \lambdadec}} - \ln {\mathcal{Z}^{-}_{n; \lambdadec}}\right) &= \lim_{n\to\infty} \frac{1}{|\calW_n|}\sum_{i\in \mathcal{W}_n^\prime} \sum_{\ell=1}^{n} \frac{1}{\ell^\delta}\left(\langle \sigma_{i+\ell e_d} \rangle^{+,J}_{n;\lambdadec} - \langle \sigma_{i+\ell e_d} \rangle^{-,J}_{n;\lambdadec}\right) \\
    &= \sum_{\ell=1}^\infty \frac{1}{\ell^{\delta}}\left( \langle \sigma_{\ell e_d} \rangle^+_{J, \lambdadec} - \langle \sigma_{\ell e_d} \rangle^-_{J, \lambdadec} \right).
\end{align*}
By FKG, all the terms in the sum are non-negative, therefore $\langle \sigma_{\ell e_d} \rangle^+_{J, \lambdadec}= \langle \sigma_{\ell e_d} \rangle^-_{J, \lambdadec}$ for all $\ell\geq 1$. By translation invariance, we conclude that $\langle \sigma_{i} \rangle^+_{J, \lambdadec} = \langle \sigma_{i} \rangle^-_{J, \lambdadec}$ for all $i\in\H+$, what shows uniqueness. 
\end{proof}

\section{Uniqueness for \texorpdfstring{$\delta<1$}{d<1}}In this section, we will prove the uniqueness of the semi-infinite Ising model with external field $\lambdadec$ with $\delta<1$, for any inverse temperature $\beta>0$ and ferromagnetic interaction. We again fix $\beta=1$ to simplify the notation. For every $k\in\Z$, let $L_k\coloneqq \Z^{d-1}\times \{k\}$ the the layer at weight $k$. We first prove uniqueness for the Ising model in $\Z^d$ with external field $\bm{h}^*$ given by \eqref{particular.external.field} and interaction $\bm{J}_{\lambda}=(J_{i,j})_{i,j\in\Z^d}$ given by
\begin{equation}\label{Eq: Definition_of_J_Lambda}
    (\bm{J}_{\lambda})_{i,j} = \begin{cases}
                                            \frac{\lambda}{2} &\text{ if }i\in L_0 \text{ and }  j\in L_{-1}\cup L_{1},\\
                                            J &\text{ otherwise.}
                                        \end{cases}
\end{equation}
for $|i-j|=1$. As we are always considering short-range interactions, $J_{i,j}=0$ whenever $|i-j|\neq 1$. After that, we show how uniqueness for this model implies uniqueness for our model of interest. 

The proof of uniqueness given by \cite{Bissacot_Cass_Cio_Pres_15} together with \cite{Cioletti_Vila_2016} only considers constant interactions. The extension to the interaction $\bm{J}_\lambda$ is a direct consequence of the monotonicity properties of the random cluster representation, proved first in \cite{Biskup_Borgs_Chayes_Kotecky_00} for constant external fields and extended in \cite{Cioletti_Vila_2016} to more general models. 

\subsection{Random Cluster Representation and Edward-Sokal coupling}
    
In this section, following \cite{Biskup_Borgs_Chayes_Kotecky_00} and \cite{Cioletti_Vila_2016}, we define the Random Cluster model (RC) and the Edward-Sokal (ES) coupling between the RC model and the Ising model. Then we show that uniqueness for the ES model implies uniqueness for the Ising model and we introduce the monotonicity property of the RC model we are interested in.

In \cite{Biskup_Borgs_Chayes_Kotecky_00} and \cite{Cioletti_Vila_2016}, they consider the Potts model and the General Random Cluster model, so their setting is more general.  We will restrict the results presented here to a particular case of interest.  

\subsubsection{The Random Cluster model}
    Given $\E = \{ \{i,j\}\subset\Z^d : \ |i-j|=1\}$, $(\Z^d,\E)$ defines a graph. The configuration space of the RC model is $\{0,1\}^\E$. A general configuration will be denoted $\omega$ and called an \textit{edge configuration}. An edge $e\in\E$ is \textit{open} (in a configuration $\omega$) if $\omega_e=1$, and it is \textit{closed} otherwise. A path $(e_0,e_1,\dots,e_n)$ is an \textit{open path} if $\omega_{e_k} = 1$ for all $k=0,\dots,n$. Vertices $i,j\in\Z^d$ are \textit{connected in} $\omega$ if there is an open path $(e_0,\dots,e_n)$ connecting $i$ and $j$, that is, $i\in e_0$ and $j\in e_n$. We denote $x\longleftrightarrow y$ when $x$ and $y$ are connected in $\omega$. The open connected component of $x\in\Z^d$ is $C_x(\omega) = \{\{i,j\}\in\E : x \longleftrightarrow i\}\cup\{x\}$. An arbitrary connected component of $\omega$ is denoted $C(\omega)$. Moreover, for any $E\subset \E$, the vertices touched by $E$ are $\V(E) = \{x\in\Z^d :  x\in e \textit{ for some }e\in E\}$. 

    Given $\Lambda\Subset \Z^d$, consider $E(\Lambda) = \{\{i,j\}\in\E : i\in\Lambda \}$ the edges with at least one endpoint in $\Lambda$. Consider also $E_0(\Lambda) = \{e\in \E : e\subset \Lambda\}$, the edges with both endpoints in $\Lambda$. For any $G=(V,E)$ finite sub-graph of $(\Z^d,\E)$, the probability measure of the Random Cluster model in $E\Subset \E$ with ferromagnetic interaction $J = (J_{i, j})_{i,j}$, external field $\bm{h}=(h)_{x\in\Z^d}$ and boundary condition $\omega_{E^c}\in\{0,1\}^{E^c}$ is 
    \begin{equation}
    \phi_{E;\bm{J}, \bm{h}}(\omega_E|\omega_{E^c}) = \frac{1}{Z^{RC, \omega_{E^c}}_{E; \bm{J},\bm{h}}}B_{\bm{J}}(\omega_E)\prod_{\substack{C(\omega): C(\omega)\cap \V(E) \neq \emptyset}}(1+ e^{-2\sum_{i\in C(\omega)}h_i}) ,     
    \end{equation}
    where the product is taken over connected open clusters only, with the convention that $e^{-\infty}=0$. The term in the denominator is the usual partition function
    \begin{equation*}
        Z^{RC, \omega_{E^c}}_{E; \bm{J},\bm{h}} = \sum_{\omega_E\in\{0,1\}^E} B_{\bm{J}}(\omega_E)\prod_{\substack{C(\omega): C(\omega)\cap \V(E) \neq \emptyset}}(1+ e^{-2\sum_{i\in C(\omega)}h_i}).
    \end{equation*}
and $B_{\bm{J}}$ is the Bernoulli like factor
    \begin{equation*}
        B_{\bm{J}}(\omega) \coloneqq \prod_{e:\omega_e=1}(e^{2 J_e} -1).
    \end{equation*}
This is not a Bernoulli factor since the weights can be bigger than one. Moreover, the interaction of an edge $e=\{i,j\}$ is, as expected, $J_e\coloneqq J_{i,j}$. For $i=0,1$,  let $\omega_{E}^{(i)}$ be the configuration satisfying $\omega_e^{(i)}=i$ for all $e\in E^c$. Two particularly important measures are the \textit{RC model with free boundary condition} in $\Lambda\Subset\Z^d$, given by
\begin{equation*}
    \phi_{\Lambda; \bm{J}, \bm{h}}^{0} \coloneqq \phi_{E_0(\Lambda); \bm{J}, \bm{h}}(\omega_{E_0(\Lambda)}|\omega_{E_0(\Lambda)}^{(0)}),
\end{equation*}
and the \textit{RC model with wired boundary condition} in $\Lambda\Subset\Z^d$, given by
\begin{equation*}
     \phi_{\Lambda; \bm{J}, \bm{h}}^{1} \coloneqq \phi_{E_0(\Lambda); \bm{J}, \bm{h}}(\omega_{E_0(\Lambda)}|\omega_{E_0(\Lambda)}^{(1)}).
\end{equation*}
The RC model is related to the Ising model through the Edwards-Sokal coupling, introduced next.
\subsubsection{The Edwards-Sokal model}

Given $\Lambda\Subset\Z^d$ and $E\Subset\E$, two configurations $\sigma\in\Omega$, $\omega\in\{0,1\}^\E$ and weights
\begin{equation*}
    \mathcal{W}(\sigma_\Lambda, \omega_E | \sigma_{\Lambda^c}, \omega_{E^c}) = \prod_{\substack{\{i,j\}\in E: \\ \omega_{i,j}=1}} \delta_{\sigma_i,\sigma_j}(e^{2 J_{i,j}}-1)\prod_{ i\in \Lambda}e^{ h_i\sigma_i},
\end{equation*}
the Edwards-Sokal (ES) measure in $\Lambda\Subset\Z^d$ and $E\Subset \E$ is given by
\begin{equation*}
    \phi^{ES}_{\Lambda, E; \bm{J}, \bm{h}}(\sigma_\Lambda, \omega_E | \sigma_{\Lambda^c}, \omega_{E^c}) \coloneqq \frac{\mathcal{W}(\sigma_\Lambda, \omega_E | \sigma_{\Lambda^c}, \omega_{E^c})}{Z^{ES}_{\Lambda, E; \bm{J}, \bm{h}}(\sigma_{\Lambda^c}, \omega_{E^c})},
\end{equation*}
with 
\begin{equation*}
    Z^{ES}_{\Lambda, E; \bm{J}, \bm{h}}(\sigma_{\Lambda^c}, \omega_{E^c}) \coloneqq \sum_{\substack{\eta_\Lambda\in\Omega_\Lambda \\ \xi_E\in \{0,1\}^E}}\mathcal{W}(\eta_\Lambda, \xi_E | \sigma_{\Lambda^c}, \omega_{E^c}).
\end{equation*}
If $Z^{ES}_{\Lambda, E; \bm{J}, \bm{h}}(\sigma_{\Lambda^c}, \omega_{E^c}) = 0$, we simply take $\phi^{ES}_{\Lambda, E; \bm{J}, \bm{h}}(\cdot | \sigma_{\Lambda^c}, \omega_{E^c}) = 0$. To simplify the notation, as we are considering arbitrary interactions and external fields, we will omit them from the notation. 

\begin{remark}
    The measures $\phi^{ES}_{\Lambda, E(\Lambda)}(\cdot|\sigma_{\Lambda^c}, \omega_{E^{c}(\Lambda)})$ do not depend on the choice of configuration $\omega_{E^{c}(\Lambda)}$. We choose to keep it in the notation since, later on, we will want to see $\phi^{ES}_{\Lambda, E(\Lambda)}$ as a specification. 
\end{remark}

The ES is indeed a coupling between the Ising and the RC model, see \cite[Theorem 1]{Cioletti_Vila_2016}. To define the infinity volume measures, we use the DLR equations. Let $\mathcal{F}_1$ (respectively $\mathcal{F}_2$) be the sigma-algebra generated by the cylinders sets on $\{0,1\}^\E$ (respectively on $\Omega\times \{0,1\}^\E$). We take $\mathcal{P}(\{0,1\}^\E)$ the set of probability measures in $(\{0,1\}^\E, \mathcal{F}_1)$ and $\mathcal{P}(\Omega\times\{0,1\}^\E)$ the set of probability measures in $(\Omega\times\{0,1\}^\E, \mathcal{F}_2)$. The set of RC measures is 
\begin{equation*}
    \mathcal{G}^{RC}_{\bm{J},\bm{h}} \coloneqq \left\{\phi\in \mathcal{P}(\{0,1\}^\E) : \phi(f) = \int \phi_E(f|\omega_E^c)\}d\phi(\omega), \text{ whenever } \supp(f) \subset E \text{ and } E\Subset \E\right\}. 
\end{equation*}
Analogously, the set of ES measures is 
\begin{multline}
    \mathcal{G}^{ES}_{\bm{J}, \bm{h}} \coloneqq \left\{ \nu\in \mathcal{P}(\Omega\times\{0,1\}^\E) :  \nu(f) = \int \phi^{ES}_{\Lambda, E(\Lambda)}(f|\sigma_\Lambda^c, \omega_{E(\Lambda)^c})d\nu(\sigma, \omega), \right. \\ \left. \text{ whenever } \supp(f) \subset \Lambda\times E(\Lambda) \text{ and } \Lambda\Subset \Z^d \right\}.
\end{multline}
We will often omit the parameter $\bm{h}$ in statements that hold for an arbitrary choice of external field.
\begin{remark}
    Since the families $\{\phi_E\}_{E\Subset\E}$ and $\{\phi^{ES}_{\Lambda, E(\Lambda)}\}_{\Lambda\Subset\Z^d}$ are specifications, the sets $\mathcal{G}^{RC}_{\bm{J}}$ and $\mathcal{G}^{ES}_{\bm{J}}$ are the usual set of DLR Gibbs measures. 
\end{remark}

The spin marginal measure of an infinity ES measure is a Gibbs measure in $\mathcal{G}_{\bm{J}}$. In fact, an even stronger statement holds. The following theorem was proved in \cite{Biskup_Borgs_Chayes_Kotecky_00} and extended to general external fields in \cite{Cioletti_Vila_2016}.

\begin{theorem}\label{Theo: ES_to_Ising_isomosphism}
Let $\Pi_S: \mathcal{G}^{ES}_{\bm{J}}: \longrightarrow \mathcal{G}^{IS}_{\bm{J}}$ be the application that takes an ES - measure to its spin marginal measure, that is, for any $\nu\in\mathcal{G}^{ES}_{\bm{J}}$ and $f:\Omega\longrightarrow \mathbb{R}$ with $\supp(f)\Subset \Z^d$, 
\begin{equation*}
    \Pi_S(\nu)(f) \coloneqq \int f(\sigma) d\nu(\sigma,\omega).
\end{equation*}
Then, $\Pi_S$ is a linear isomorphism. In particular, $|\mathcal{G}^{ES}_{\bm{J}}| = 1$ if and only if $|\mathcal{G}_{\bm{J}}^{IS}|=1$. 
\end{theorem}

This shows that uniqueness for the ES model implies uniqueness for the Ising model. It is possible to relate the uniqueness of the RC model with the uniqueness of the ES model. To do so, we use the FKG property, and some consequences of it, of the RC and ES models. The main contribution of \cite{Cioletti_Vila_2016} was the extension of these properties from the models with constant external fields to models with non-constant external fields. These results are described next.

 The key contribution of \cite{Cioletti_Vila_2016} was to prove that the RC model with space-dependent external fields has the \textit{strong FKG} property, see \cite[Theorem 4]{Cioletti_Vila_2016}. Two consequences of the strong FKG property we will use are the extremality of the free and wired states \cite[Theorem 6]{Cioletti_Vila_2016} and the monotonicity of these states with respect to the volume \cite[Lemma 9]{Cioletti_Vila_2016}.
As we did for the configuration space, we can consider a partial order on $\{0,1\}^\E$ defining $\omega\leq \omega^\prime$ when $\omega_e\leq\omega^\prime_e$ for all $e\in\E$, what also defines increasing function. The next result allows us to compare the models with interaction $\bm{J}_\lambda$ and constant interaction $\bm{J}\equiv J$. The proof is a straightforward adaptation of \cite[Theorem 7]{Cioletti_Vila_2016}.

\begin{proposition}\label{Prop: RC_is_increasing_in_J}
   Let $\bm{J}=\{J_{i,j}\}_{i,j\in \Z^d}$ and $\bm{J}^\prime=\{J^\prime_{i,j}\}_{i,j\in \Z^d}$ be nearest-neighbor interactions with $0\leq J_{i,j}\leq J^\prime_{i,j}$, for all $i,j\in\Z^d$. Then, for any $\Lambda\Subset\Z^d$ and $f$ local non-decreasing function,
   \begin{equation*}
       \phi^0_{\Lambda; \bm{J}}(f)\leq \phi^0_{\Lambda; \bm{J}^\prime}(f) \qquad \text{ and } \qquad   \phi^1_{\Lambda; \bm{J}}(f)\leq \phi^1_{\Lambda; \bm{J}^\prime}(f).
   \end{equation*}
\end{proposition}
\begin{proof}
    Consider a function $g:\{0,1\}^{E(\Lambda)}\longrightarrow \mathbb{R}$ given by
    \begin{equation*}
    g(\omega)= \prod_{e\in E(\Lambda)}\left(\frac{e^{2 J_e}-1}{e^{2 J^\prime_e} - 1}\right)^{\omega_e}.
    \end{equation*}
    By the restriction on $\bm{J}$ and $\bm{J}^\prime$, all the fractions above are at most $1$, so $g$ is non-increasing. Given a non-decreasing local function $f$, 
    \begin{equation*}
        \phi^0_{\Lambda;\bm{J}}(f) = \frac{1}{Z^{0}_{\Lambda;\bm{J}}}\sum_{\omega\in\{0,1\}^{E_0(\Lambda)}}f(\omega)g(\omega)\prod_{e:\omega_e=1}\left(e^{2 J_e^\prime}-1\right)\prod_{C(\omega)}\left(1+e^{-2\sum_{i\in C(\omega)}h_i}\right) =  \frac{Z^{0}_{\Lambda;\bm{J}^\prime}}{Z^{0}_{\Lambda;\bm{J}}}\phi^0_{\Lambda;\bm{J}^\prime}(f.g).
    \end{equation*}
    Taking, in particular, $f\equiv 1$, we get $\phi^0_{\Lambda;\bm{J}^\prime}(g) = \frac{Z^{0}_{\Lambda;\bm{J}}}{Z^{0}_{\Lambda;\bm{J}^\prime}}$. Using the FKG property, we conclude that 
    \begin{equation*}
        \phi^0_{\Lambda;\bm{J}}(f)  = \frac{\phi^0_{\Lambda;\bm{J}^\prime}(f.g)}{\phi^0_{\Lambda;\bm{J}^\prime}(g)}\leq \phi^0_{\Lambda;\bm{J}^\prime}(f).
    \end{equation*}
This same argument can be made for the wired boundary condition, which concludes the proof.
\end{proof}

To guarantee uniqueness for the RC model, we can use the quantity
\begin{equation*}
    P_\infty(\bm{J}, \bm{h}) \coloneqq \sup_{x\in\Z^d}\sup_{\phi\in \mathcal{G}^{RC}}\phi\left(|C_x| = +\infty \right),
\end{equation*}
together with the next theorem, proved in \cite{Cioletti_Vila_2016}.
\begin{theorem}\label{Theo: Uniqueness_with_P_infinity}
    For any ferromagnetic nearest-neighbor interaction $\bm{J}=\{J_{i,j}\}_{i,j\in\Z^d}$ and non-negative external field $\bm{h}=(h_{i})_{i\in\Z^d}$, if $P_\infty(\bm{J}, \bm{h})=0$, then $\left|\mathcal{G}^{ES}_{J}\right|=\left|\mathcal{G}^{RC}_{J}\right|=1$.
\end{theorem}

\subsection{Proof of Uniqueness when \texorpdfstring{$\delta<1$}{d<1}}
To prove uniqueness for the semi-infinite Ising model, we first prove uniqueness for the usual Ising model with interaction given by \eqref{Eq: Definition_of_J_Lambda} using the RC and ES models.

\begin{theorem}\label{Theo: Uniquiness_Ising_J_lambda}
    The Ising model in $\Z^d$ with interaction $\bm{J}_\lambda$ defined in \eqref{Eq: Definition_of_J_Lambda}, $0\leq\lambda\leq 2J$ and external field $\bm{h}^*=(h_i^*)_{i\in\Z^d}$ given by \eqref{particular.external.field} has a unique state. 
\end{theorem}
\begin{proof}
    By Theorem \ref{Theo: ES_to_Ising_isomosphism}, it is enough to show that $\left| \mathcal{G}^{ES}_{ \bm{J}_\lambda, \bm{h}^*}\right|=1$. It was shown in \cite{Cioletti_Vila_2016} that, for any constant nearest-neighbor ferromagnetic interaction $\bm{J}\equiv J$, $P_\infty(\bm{J}, \bm{h}^*)=0$ and, in particular, $\phi^1_{\bm{J}, \bm{h}^*}\left( |C_x| = +\infty \right) = 0$ for all $x\in \Z^d$. 
    For any $x\in\Z^d$, the function $\mathbbm{1}_{\{|C_x| = +\infty\}}$ is increasing. Then, Proposition \ref{Prop: RC_is_increasing_in_J} yields
    \begin{equation*}
        \phi^1_{\Lambda; \bm{J}_\lambda, \bm{h}^*}\left( |C_x| = +\infty \right) \leq    \phi^1_{\Lambda; \bm{J}, \bm{h}^*}\left( |C_x| = +\infty \right)
    \end{equation*}
    for any $\Lambda\Subset\Z^d$ and $x\in\Z^d$. By the monotonicity of the wired stated with respect to the volume $\Lambda$, we can take the limit $\Lambda\nearrow\Z^d$ to get $ \phi^1_{\bm{J}_\lambda, \bm{h}^*}\left( |C_x| = +\infty \right) =0$ for all $x\in\Z^d$. Since $\phi^1_{\bm{J}_\lambda, \bm{h}^*}$ is extremal we conclude that $P_\infty(\bm{J}_\lambda, \bm{h}^*)=0$, and therefore we have $\left| \mathcal{G}^{ES}_{\bm{J}_\lambda, \bm{h}^*}\right|=1$ by Theorem \ref{Theo: Uniqueness_with_P_infinity}. 
\end{proof}

Now we prove the main result of this section. We prove uniqueness for the semi-infinite Ising with external field given by \eqref{particular.external.field} and $\delta<1$ at any temperature, by comparing it with the model of Theorem \ref{Theo: Uniquiness_Ising_J_lambda}.
\begin{theorem}
    The semi-infinite Ising model with interaction $J>0$ and external field $\lambdadec=(\lambda_i)_{i\in\H+}$ with $\lambda_i=\frac{\lambda}{i_d^\delta}$
for all $i\in\H+$ and $\delta<1$ has a unique Gibbs state. 
\end{theorem}

\begin{proof}
        Proposition \ref{basta.comparar.magnetizacoes} guarantees that it is enough to prove $\langle \sigma_{e_d} \rangle^{+}_{\bm{\lambda}} = \langle \sigma_{e_d} \rangle^{-}_{\bm{\lambda}}$, where $e_d=(0,\dots,0,1)$ is a base vector of $\H+$. By spin-flip symmetry, we can assume without loss of generality that $\lambda\geq 0$. Split $\mathbb{Z}^d$ in layers $L_k \coloneqq \mathbb{Z}^{d-1}\times\{k\}$, with $k\in\mathbb{Z}$. Lets first consider the case $0\leq\lambda \leq 2J$. For any $\Lambda_n=[-n,n]^{d-1}\times[1, n]$, we rewrite the semi-infinite model as the standard Ising model but now with interaction $\bm{J}_{\lambda}$ and external field $\bm{\lambda}_0$ given by 
    $$(\bm{\lambda}_0)_i = \begin{cases} 
                        \lambda/2 &\text{ if } i\in L_1,\\
                        \lambda\mid i_d \mid^{-\delta}, &\text{ if } i\in L_k, k>1,\\
                        0 &\text{ otherwise}, \\
                        \end{cases}$$ 
    for any $i,j\in\mathbb{Z}^d$. Hence, 
    \begin{align}
        &\langle \sigma_{e_d} \rangle^{+}_{\Lambda_n; \bm{\lambda}} = \langle \sigma_{e_d} \rangle^{+, \bm{J}_{\lambda}}_{\Lambda_n; \bm{\lambda}_0} &\text{ and } &&\langle \sigma_{e_d} \rangle^{-}_{\Lambda_n; \bm{\lambda}} = \langle \sigma_{e_d} \rangle^{\mp, \bm{J}_{\lambda}}_{\Lambda_n; \bm{\lambda}_0}
    \end{align}
where $(\mp)_i = \mathbbm{1}_{\{i\in\mathbb{Z}^d\setminus \H+\}} -  \mathbbm{1}_{\{i\in\H+\}}$. Consider $\bm{\lambda}_0^\prime$ an extension of $\bm{\lambda}_0$ to $\Z^d$ given by $(\lambda_0^\prime)_{i} \coloneqq (\lambda_0)_{i}$ when $i\in \H+$ and $(\lambda_0^\prime)_{i} \coloneqq (\lambda_0)_{i^\prime}$ when $i\in \mathbb{Z}^d\setminus\H+$, where $i^\prime = (i_1,\dots,i_{d-1}, -i_d)$. Since the boxes $\Lambda_n$ and $\Lambda_n^\prime \coloneqq [-n,n]^{d-1}\times [-n, -1]$ are not connected, taking $\Delta_n^\prime\coloneqq \Lambda_n \cup \Lambda_n^\prime$ we have 

\begin{align}
    &\langle \sigma_{e_d} \rangle^{+, \bm{J}_{\lambda}}_{\Lambda_n; \bm{\lambda}_0} = \langle \sigma_{e_d} \rangle^{+, \bm{J}_{\lambda}}_{\Delta_n^\prime; \bm{\lambda}^\prime_0} &\text{ and } && \langle \sigma_{e_d} \rangle^{\mp, \bm{J}_{\lambda}}_{\Lambda_n; \bm{\lambda}_0} = \langle \sigma_{e_d} \rangle^{\mp, \bm{J}_{\lambda}}_{\Delta_n^\prime; \bm{\lambda}^\prime_0}.
\end{align}
For every $h\geq0$, let $\bm{h}^w$ be an external field acting only on $L_0$, that is  $h^w_i=h\mathbbm{1}_{\{i\in L_0\}}$ for all $i\in\Z^d$. Then, denoting $\Delta_n=\Delta_n^\prime\cup L_0$,

\begin{align}
    &\langle \sigma_{e_d} \rangle^{+, \bm{J}_{\lambda}}_{\Delta_n^\prime; \bm{\lambda}^\prime_0} = \lim_{h\to\infty}\langle \sigma_{e_d} \rangle^{+, \bm{J}_{\lambda}}_{\Delta_n; \bm{\lambda}^\prime_0 + \bm{h}^w} &\text{ and } &&  \langle \sigma_{e_d} \rangle^{\mp, \bm{J}_{\lambda}}_{\Delta_n^\prime; \bm{\lambda}^\prime_0} = \lim_{h\to \infty}  \langle \sigma_{e_d} \rangle^{\mp, \bm{J}_{\lambda}}_{\Delta_n; \bm{\lambda}^\prime_0 + \bm{h}^w}.
\end{align}
Differentiating in $h$ and using Proposition  \ref{Consequence.of.DVI}, we see that the difference $\langle \sigma_{e_d} \rangle^{+, \bm{J}_{\lambda}}_{\Delta_n; \bm{\lambda}^\prime_0 + \bm{h}^0} - \langle \sigma_{e_d} \rangle^{\mp, \bm{J}_{\lambda}}_{\Delta_n; \bm{\lambda}^\prime_0 + \bm{h}^0}$ is decreasing in $h$ for $h\geq 0$. Proposition \ref{Consequence.of.DVI} was stated for the semi-infinite states, but it also holds for Ising states since the duplicated variables inequalities are in this generality. We can then bound the difference of the states by choosing the particular case $h=\lambda$. Denoting $\bm{\lambda}_l = \bm{\lambda}^\prime_0 + \bm{\lambda}^0$, we conclude that 
\begin{equation}\label{eq: uniqueness}
    \langle \sigma_{e_d} \rangle^{+}_{\Lambda_n; \bm{\lambda}}- \langle \sigma_{e_d} \rangle^{-}_{\Lambda_n; \bm{\lambda}}\leq \langle \sigma_{e_d} \rangle^{+, \bm{J}_{\lambda}}_{\Delta_n; \bm{\lambda}_l} - \langle \sigma_{e_d} \rangle^{\mp, \bm{J}_{\lambda}}_{\Delta_n; \bm{\lambda}_l}.
\end{equation}
Again by Proposition \ref{Consequence.of.DVI}, the RHS of equation \eqref{eq: uniqueness} in non-increasing in $h_i$, the external field on the site $i$, for any $i\in\mathbb{Z}^d$. So, denoting $\bm{\lambda}^{Is}$ the external field \eqref{particular.external.field} considered in \cite{Bissacot_Cass_Cio_Pres_15} with $h^*=\lambda/2$, we have that, for any $i\in\mathbb{Z}^d$, $h_i^{Is}=\frac{\lambda}{2}|i|^{-\delta}\leq \frac{\lambda}{2}|i_d|^{\delta}\leq (\bm{\lambda}_l)_i$, hence 
\begin{equation}
    \langle \sigma_{e_d} \rangle^{+}_{\Lambda_n; \bm{\lambda}}- \langle \sigma_{e_d} \rangle^{-}_{\Lambda_n; \bm{\lambda}}\leq \langle \sigma_{e_d} \rangle^{+, \bm{J}_{\lambda}}_{\Delta_n; \bm{\lambda}^{Is}} - \langle \sigma_{e_d} \rangle^{\mp, \bm{J}_{\lambda}}_{\Delta_n; \bm{\lambda}^{Is}}.
\end{equation}
Taking the limit in $n$, the RHS of equation above goes to $\langle \sigma_{e_d} \rangle^{+, \bm{J}_{\lambda}}_{\bm{\lambda}^{Is}} - \langle \sigma_{e_d} \rangle^{\mp, \bm{J}_{\lambda}}_{\bm{\lambda}^{Is}}$, that is equal to $0$ by Theorem \ref{Theo: Uniquiness_Ising_J_lambda}. We conclude that $\langle \sigma_{e_d} \rangle^{+}_{\bm{\lambda}} - \langle \sigma_{e_d} \rangle^{-}_{\bm{\lambda}}=0$, so there is only one Gibbs state by Proposition \ref{basta.comparar.magnetizacoes}.

When $\lambda > 2J$, we replace $\bm{J}_\lambda$ by $\bm{J}\equiv J$ in the argument above and the same proof holds with minor adjustments. 
\end{proof}

\section{Concluding Remarks}

This paper studied the semi-infinite Ising model with an external field of the form $h_i = \lambda|i_d|^{-\delta}$, where $\lambda$ is the wall influence. We proved the existence of a critical value $\overline{\lambda}_c$ for which there is phase transition for $0\leq \lambda<\overline{\lambda}_c$ and there is uniqueness for $\lambda>\overline{\lambda}_c$. Moreover, we proved that when $\delta<1$, $\overline{\lambda}_c=0$, there is uniqueness for all temperatures and all wall parameters.

The external field $\lambdadec$ resembles the long-range Ising model interaction, given by $J_{x,y}=|x-y|^{\alpha}$, with $\alpha>d$. One natural question is if our results can be extended to a long-range semi-infinite Ising model. Phase transition for the long-range model with the decaying field was proved in \cite{Affonso.2021}, but there is no uniqueness argument for $\delta<1$. We believe the case $\delta>1$ should extend with little effort. However, the proof of uniqueness when $\delta<1$ should be far from trivial since it is already the case for short-range models. 

Our methods are not suited to treat the critical exponent $\delta=1$, since the wall-free energy may be ill-defined. In this critical region, the Ising model with external field \eqref{particular.external.field} presents phase transition when $h^*$ is small enough. It is also expected that the model presents uniqueness for $h^*$ large enough, but no proof is available. This indicates that the semi-infinite model with $\delta=1$ should also have a critical wall influence that dictates phase transition.  

The semi-infinite model has a second type of phase transition, called \textit{layering transition}, which occurs when the wall influence and the external field have opposite signs. Some results about this transition can be found in \cite{FP-II}. There are no results on this kind of transition when we consider decaying fields. 

In recent years, there have been a lot of developments in the random field Ising model, which consists of the usual Ising model with a random external field. There are no rigorous results for the semi-infinite model with random perturbation. The random field disrupts phase transition in the two-dimensional Ising model. The same is expected to happen in the semi-infinite model. In the semi-infinite model, we can also consider a random wall influence $\lambda$, that is site dependent, and study if this affects the phase transition. 

\section*{Acknowledgements}

RB is supported by CNPq grant 408851/2018-0 and by FAPESP grant 16/25053-8. JM is supported by FAPESP Grant 2018/26698-8 and 2016/25053-8. JM is very grateful to Yvan Velenik for fruitful discussions while participating in the NCCR SwissMAP Masterclass program 2019/2020 during his stay at the Universit\'e de Gen\`eve.

\bibliographystyle{habbrv} 
\bibliography{main}  

\end{document}